\documentclass[12pt]{article}
\usepackage[margin=1in]{geometry}

\title{Ideal Membership Problem for Boolean Minority}

\makeatletter
\renewcommand\@date{{%
  \vspace{-\baselineskip}%
  \large\centering
  \begin{tabular}{@{}c@{}}
    Arpitha P. Bharathi\textsuperscript{1} \\
    \normalsize arpitha@idsia.ch
  \end{tabular}%
  \quad and\quad
  \begin{tabular}{@{}c@{}}
    Monaldo Mastrolilli\textsuperscript{2} \\
    \normalsize monaldo@idsia.ch
  \end{tabular}

  \bigskip
  
  \textsuperscript{1,}\textsuperscript{2}IDSIA-SUPSI, Switzerland

  \bigskip

}}
\makeatother

\usepackage{amsthm}
\usepackage{amsmath}
\usepackage{mathrsfs}
\usepackage{amssymb}
\usepackage{verbatim}
\usepackage{algorithmic}
\usepackage{algorithm2e}
\usepackage{color}
\usepackage{enumerate}
\usepackage{mdframed}
\usepackage{marginnote}
\usepackage{amsfonts}
\usepackage{lineno}

\newcommand\mydef{\mathrel{\overset{\makebox[0pt]{\mbox{\normalfont\tiny\sffamily def}}}{=}}}

\newcommand{\Mod}[1]{\ (\mathrm{mod}\ #1)}

\newcommand{\sos}{\text{\sc{SoS}}}

\newcommand{\Cc}{\mathcal{C}}

\newcommand{\Ideal}[1]{{\textbf{I}}\left( #1 \right)}
\newcommand{\GIdeal}[1]{\left\langle #1 \right\rangle}

\newcommand{\CSP}{\textsc{CSP}}
\newcommand{\IMP}{\textsc{IMP}}

\newcommand{\Minority}{\textsf{Minority }}
\newcommand{\Minorityns}{\textsf{Minority}}
\newcommand{\lex}{\textsf{lex }}

\newcommand{\grlex}{\textsf{grlex }}

\newcommand{\Variety}[1]{{\textbf{V}}\left( #1 \right)}
\newcommand{\spn}[1]{\left\langle #1 \right\rangle}

\newcommand{\x}{\mathbf{x}}
\newcommand{\I}{\emph{\texttt{I}}}
\newcommand{\Zz}{\mathbb{Z}}
\newcommand{\N}{\mathbb{N}}

\newcommand{\multideg}{\textnormal{multideg}}
\newcommand{\LM}{\textnormal{LM}}
\newcommand{\LT}{\textnormal{LT}}
\newcommand{\LC}{\textnormal{LC}}
\newcommand{\LCM}{\textnormal{lcm}}
\newcommand{\GB}{\text{Gr\"{o}bner} }
\newcommand{\reduce}[2]{{#1}|_{#2}}

\newcommand{\Field}{\mathbb{F}}
\newcommand{\Real}{\mathbb{R}}

\newcommand{\vb}[1]{\mathbf{#1}}

\usepackage{hyperref} \hypersetup{
pdftitle={Shown in AR File Information}, pdfstartview=FitH, 
} 

\usepackage[capitalize]{cleveref}

\newtheorem{theorem}{Theorem}[section]
\newtheorem{lemma}[theorem]{Lemma}

\newtheorem{corollary}[theorem]{Corollary}

\newtheorem{proposition}[theorem]{Proposition}

\newtheorem{example}{Example}[section]
\newtheorem{definition}{Definition}[section]

\begin{document}

\maketitle


%

\begin{abstract}
The Ideal Membership Problem (IMP) tests if an input polynomial $f\in \Field[x_1,\dots,x_n]$ with coefficients from a field $\Field$ belongs to a given ideal $I \subseteq \Field[x_1,\dots,x_n]$. It is a well-known fundamental problem with many important applications, though notoriously intractable in the general case. In this paper we consider the IMP for polynomial ideals encoding combinatorial problems and where the input polynomial $f$ has degree at most $d=O(1)$ (we call this problem IMP$_d$).  
  
A dichotomy result between ``hard'' (NP-hard) and ``easy'' (polynomial time) $\IMP$s was recently achieved for Constraint Satisfaction Problems over finite domains~\cite{Bulatov17,Zhuk17} (this is equivalent to  $\IMP_0$) and $\IMP_d$ for the Boolean domain~\cite{MonaldoMastrolilli2019}, both based on the classification of the $\IMP$ through functions called polymorphisms. For the latter result, there are only six polymorphisms to be studied in order to achieve a full dichotomy result for the $\IMP_d$. The complexity of the $\IMP_d$ for five of these polymorphisms has been solved in \cite{MonaldoMastrolilli2019} whereas for the ternary minority polymorphism it was incorrectly declared in \cite{MonaldoMastrolilli2019} to have been resolved by a previous result. As a matter of fact the complexity of the $\IMP_d$ for the ternary minority polymorphism is open.

In this paper we provide the missing link by proving that the IMP$_d$ for Boolean combinatorial ideals whose constraints are closed under the minority polymorphism can be solved in polynomial time. 

This is achieved by first showing that a \GB basis can be efficiently computed in the lexicographic order for these ideals. Since this is insufficient for the efficient solvability of the $\IMP_d$, we show how this \GB basis can be converted to a $d$-truncated \GB basis in graded lexicographic order in polynomial time which ensures the achievement of the result. This result, along with the results in \cite{MonaldoMastrolilli2019}, completes the identification of the precise borderline of tractability for the $\IMP_d$ for constrained problems over the Boolean domain.  

This paper is motivated by the pursuit of understanding the recently raised issue of bit complexity of Sum-of-Squares proofs raised by O'Donnell~\cite{odonnell2017}. Raghavendra and Weitz \cite{raghavendra_weitz2017} show
how the $\IMP_d$ tractability for combinatorial ideals implies bounded coefficients in Sum-of-Squares proofs.
\end{abstract}
\section{Introduction}
A polynomial ideal is a subset of the polynomial ring $ \mathbb{F}[x_1,\dots,x_n]$ with two properties: for any two polynomials $f,g$ in the ideal, $f+g$ also belongs to the ideal and so does $hf$ for any polynomial $h$. The Hilbert Basis Theorem \cite{HilbertBasisTheorem} states that every ideal $I$ is finitely generated by a set $F=\{f_1,\dots,f_m\}\subset I$, i.e., any polynomial in $I$ is a polynomial combination of elements from $F$. The polynomial Ideal Membership Problem (IMP) is to find out if a polynomial $f$ belongs to an ideal $I$ or not, given a set of generators of the ideal. 
This fundamental algebraic complexity problem was first pioneered by David Hilbert \cite{Hilbert1893} and has important applications in solving polynomial systems and polynomial identity testing \cite{Cox,PIT}. The IMP is, in general, EXPSPACE-complete and Mayr and Meyer show that the problem for multivariate polynomials over the rationals is solvable in exponential space \cite{Mayr1989,MAYR1982305}. The IMP is intractable (can be decided in single exponential time \cite{Dickenstein1991TheMP}) even when the ideal in question is zero-dimensional (number of common zeros of generators is finite). 


The \emph{vanishing ideal} of a set $S\subseteq \mathbb{F}^n$ is the set of all polynomials in $\mathbb{F}[x_1,\dots,x_n]$ that vanish at every point of $S$. This set of polynomials forms an ideal.
In this paper we consider vanishing ideals of the sets $S$ of feasible solutions that arise from Boolean combinatorial optimization problems. The vanishing ideal of the solution space $S$ is defined as its \emph{combinatorial ideal}.
We consider the IMP for polynomial ideals encoding combinatorial problems.  We call such problems where the input polynomial $f$ has degree at most $d=O(1)$ as IMP$_d$.
The polynomial ideals that arise from combinatorial optimization problems frequently have special properties: these ideals are finite domain and therefore zero-dimensional and radical.
The question of identifying problem restrictions which are sufficient to ensure the $\IMP_d$ tractability is important from both a practical and a theoretical viewpoint, and has an immediate application to Sum-of-Squares ($\sos$) proof systems (or Lasserre relaxations) as explained in the following. 

The $\sos$ proof system is an increasingly popular tool to solve combinatorial optimization problems. Especially over the last few decades, $\sos$ has had several applications in continuous and discrete optimization (see, e.g.,~\cite{Laurent2009}). 
It was generally believed that a degree $d$ $\sos$ proof could be computed (if one existed) via the Ellipsoid algorithm in $n^{O(d)}$ time.
O'Donnell~\cite{odonnell2017}, who initially also believed this, gave a counterexample: a polynomial system and a polynomial which had degree two proofs of non-negativity with coefficients of exponential bit-complexity that forced the Ellipsoid algorithm to take exponential time. O'Donnell~\cite{odonnell2017} raised the open problem to establish useful conditions under which ``small'' $\sos$ proof can be guaranteed automatically.
A first elegant approach to this question is due to Raghavendra and Weitz~\cite{raghavendra_weitz2017} by providing a \emph{sufficient} condition on a polynomial system  that implies bounded coefficients in $\sos$ proofs.
In particular, the work of Raghavendra and Weitz \cite{raghavendra_weitz2017} shows
that the $\IMP_d$ tractability for combinatorial ideals implies polynomially bounded coefficients in $\sos$ proofs. Therefore, the $\IMP_d$ tractability yields to degree $d$ $\sos$ proof (if one exists) computation via the Ellipsoid algorithm in $n^{O(d)}$ time.
%
Hence the following question poses itself: 
Which \emph{restrictions} on combinatorial problems can guarantee an efficient computation of the $\IMP_d$?

In this paper we consider restrictions on the so-called \emph{constraint language}, namely a set of relations that is used to form the constraints of the considered combinatorial optimization problem. Each constraint language $\Gamma$ gives rise to a particular polynomial ideal membership problem, denoted $\IMP_d(\Gamma)$, and the goal is to describe the complexity of the $\IMP_d(\Gamma)$ for all constraint languages~$\Gamma$.
This kind of restrictions on the constraint languages have been successfully applied to study the computational complexity classification (and other algorithmic properties) of the decision version of Constraint Satisfaction Problems ($\CSP$) over a fixed constraint language $\Gamma$ on a finite domain, denoted $\CSP(\Gamma)$ (see \cref{sec:preliminaries}).
This classification started with the classic dichotomy  result of Schaefer \cite{Schaefer78} for 0/1 $\CSP$s, and culminated with the recent papers by Bulatov~\cite{Bulatov17} and Zhuk~\cite{Zhuk17}, settling the long-standing Feder-Vardi dichotomy conjecture for finite domain $\CSP$s. We refer to~\cite{2017dfu7} for an excellent survey. Note that $\CSP(\Gamma)$ corresponds to the very special case of the $\IMP_d(\Gamma)$ with $d=0$, i.e. where we are only interested in testing if the constant polynomial ``$1$'' belongs to the combinatorial ideal (see \cref{sect:idealCSP} for more details on Ideal-CSP correspondence). In this paper we are interested in the problem with $d\geq 1$.
%
%

Mastrolilli \cite{MonaldoMastrolilli2019} recently claimed a dichotomy result for the $\IMP_d(\Gamma)$ that fully answers the above question for 0/1 combinatorial problems: for any constant $d\geq 1$, the $\IMP_d(\Gamma)$ of Boolean combinatorial ideals is either decidable in polynomial time or it is NP-complete. 
Note that the solvability of CSP($\Gamma$) (and therefore of the $\IMP_0(\Gamma)$) in the Boolean domain is known to admit a nice dichotomy result~\cite{Schaefer78}: it is solvable in polynomial time if all constraints are closed under one of six polymorphisms (majority, minority, MIN, MAX, constant 0 and constant 1), else it is NP-complete. In \cite{MonaldoMastrolilli2019} it is claimed that the IMP$_d(\Gamma)$ for the Boolean domain also has a nice dichotomy result: it is solvable in polynomial time if all constraints are closed under one of four polymorphisms (majority, minority, MIN, MAX), else it is NP-complete. The complexity of the $\IMP_d(\Gamma)$ for five of these polymorphisms has been solved in \cite{MonaldoMastrolilli2019} whereas for the ternary minority polymorphism it was incorrectly declared in~\cite{MonaldoMastrolilli2019} to have been resolved by a previous result. As a matter of fact the complexity of the $\IMP_d(\Gamma)$ for the ternary minority polymorphism is open.

In this paper we solve this issue by providing the missing link and therefore establishing the full dichotomy result claimed in \cite{MonaldoMastrolilli2019}. To ensure efficiency of the $\IMP_d$, it is sufficient to compute a $d$-truncated \GB basis in the graded lexicographic order (see \cref{def:dTruncated GB,sec:preliminaries,sect:background} for definitions and more details).
%
This is achieved by first showing that a \GB basis can be efficiently computed in the lexicographic order for the minority polymorphism. Since this is insufficient for the efficient solvability of the $\IMP_d$, we show how this \GB basis can be converted to a $d$-truncated \GB basis in the graded lexicographic order in polynomial time. This efficiently solves the IMP$_d$ for combinatorial ideals whose constraints are over a language closed under the minority polymorphism. Together with the results in~\cite{MonaldoMastrolilli2019}, our result allows to complete the answer of the aforementioned question by allowing to identify the precise borderline of tractability of the $\IMP_d(\Gamma)$. 

Moreover, we believe the techniques described in this paper can be generalized for a finite domain with prime $p$ elements. The basis of this claim comes from the fact that constraints that are linear equations (mod $p$) are associated with an affine polymorphism \cite{Jeavons_UnifyingFramework1995}. We claim that the $\IMP_d$ is tractable for problems that are constrained as linear equations (mod $p$). The details are currently being worked out and will soon be updated in the full version of this paper. 
This is a first step towards the long term and challenging goal of generalizing the dichotomy results of solvability of the $\IMP_d$ for finite domains.


\textbf{Structure of the paper:} \cref{sec:preliminaries} contains the basic definitions required for this paper, although a reader unfamiliar with $\CSP$s over a constraint language or algebraic geometry and \GB bases is strongly recommended to read the standard literature ~\cite{Chen09,Cox} or \cref{sect:background}.
%

We concretely state our results in \cref{sec:Our Contributions}. In \cref{sec:GBlex} 
we show that the reduced \GB basis in lexicographic order can be efficiently computed for combinatorial problems constrained under the minority polymorphism.
This is achieved in \cref{sec:GBlex} by first computing a \GB basis in modular arithmetic and then transforming it into a \GB basis $G_1$ in regular arithmetic. However, this \GB basis is in the lexicographic monomial ordering, and does not guarantee the efficient solvability of the $\IMP_d$. 
In \cref{sec:conversion} we show how to convert $G_1$ to a $d$-truncated \GB basis $G_2$ in graded lexicographic monomial ordering. We prove that this conversion can be obtained in polynomial time for any fixed $d=O(1)$. A simple example is provided in \cref{sec:example}.

\subsection{Preliminaries}\label{sec:preliminaries} 
%

Let $D$ denote a finite set (\emph{domain}).
By a $k$-ary \textbf{\emph{relation}} $R$ on a domain $D$ we mean a subset of the $k$-th cartesian power $D^k$; $k$ is said to be the \emph{arity} of the relation. We often use relations and (affine) varieties interchangeably since both essentially represent a set of solutions. A \textbf{\emph{constraint language}} $\Gamma$ over $D$ is a set of relations over $D$. A constraint language is \textbf{\emph{finite}} if it contains finitely many relations, and is \emph{Boolean} if it is over the two-element domain $\{0,1\}$. In this paper, $D$ is the Boolean domain.

  A \emph{\textbf{constraint}} over a constraint language $\Gamma$ is an expression of the form $R(x_1,\ldots, x_k)$ where $R$ is a relation of arity $k$ contained in $\Gamma$, and the $x_i$ are variables. A constraint is satisfied by a mapping $\phi$ defined on the $x_i$ if $(\phi(x_1),\ldots, \phi(x_k))\in R$.

\begin{definition}\label{def:csp}
 The \emph{(nonuniform) \textsc{Constraint Satisfaction Problem} ($\CSP$)} associated with language $\Gamma$ over $D$ is the problem $\CSP(\Gamma)$ in which: an instance is a triple $ \Cc=(X,D,C)$ where $X=\{x_1,\ldots,x_n\}$ is a set of $n$ variables and $C$  is a set of constraints over $\Gamma$ with variables from $X$. The goal is to decide whether or not there exists a solution, i.e. a mapping $\phi: X\rightarrow D$ satisfying all of the constraints. We will use $Sol( \Cc)$ to denote the set of solutions of $ \Cc$.
\end{definition}
Moreover, we follow the algebraic approach to Schaefer's dichotomy result \cite{Schaefer78}  formulated by Jeavons \cite{JEAVONS1998185} where each class of CSPs that are polynomial time solvable is associated with a polymorphism.
\begin{definition}\label{def:polymorph}
An operation $f:D^m \rightarrow D$ is a \textbf{polymorphism} of a relation $R\subseteq D^k$ if for any choice of $m$ tuples from $R$ (allowing repetitions), it holds that the tuple obtained from these $m$ tuples by applying $f$ coordinate-wise is in $R$. If this is the case we also say that $f$ \emph{preserves} $R$, or that $R$ is \emph{invariant} or \emph{closed} with respect to $f$. A polymorphism of a constraint language $\Gamma$ is an operation that is a polymorphism of every $R\in \Gamma$.
\end{definition}


In this paper we deal with the \emph{minority} polymorphism: 
\begin{definition}\label{def:minority}
For a finite domain $D$, a ternary operation $f$ is called a \emph{minority} polymorphism (denoted as \Minorityns) if  $f(a,a,b)=f(a,b,a)=f(b,a,a)=b$ for all $a,b\in D$.  
\end{definition}
Note that there is only one minority polymorphism (\Minority in short) for the Boolean domain. 

\begin{example}
Consider relations $R_1=\{(0,0,1),(1,0,0),(0,1,1),(1,1,0)\}$ and $R_2=\{(1,1),(0,1)\}$ associated with language $\Gamma$ over $D=\{0,1\}$. Observe that both $R_1$ and $R_2$ are closed under \Minorityns. Consider the instance $(X=\{x,y,z\},D,C=\{C_1,C_2\})$ where constraint $C_1 = R_1 (x,y,z)$ and $C_2= R_2(x,z)$. The assignment $\phi$ where $\phi(x)=0, \phi(y)= 0, \phi(z)=1$ is a solution to this instance of CSP($\Gamma$). 
\end{example}

For a given instance $\mathcal{C}$ of CSP($\Gamma$), the \textbf{\emph{combinatorial ideal}} $\Ideal{Sol( \Cc)}$ is defined as the vanishing ideal of set $Sol( \Cc)$, (see~\cref{def:ideal} in \cref{sect:background}). We call polynomials of the form $x_i(x_i-1)$ \textit{\textbf{domain polynomials}}, denoted by $dom(x_i)$, and it is easy to see that they belong to $\Ideal{Sol( \Cc)}$ for every $i\in [n]$ as they describe the fact that $Sol( \Cc) \subseteq D^n$. For a more detailed Ideal-CSP correspondence we refer to \cref{sect:idealCSP}.

\begin{definition}\label{def:IMP}
 The {\emph{\textsc{Ideal Membership Problem}}} associated with language $\Gamma$ is the problem $\IMP(\Gamma)$ in which
 the input consists of a polynomial $f\in \Field[X]$ and a $\CSP(\Gamma)$ instance $ \Cc=(X,D,C)$. The goal is to decide whether $f$ lies in the combinatorial ideal~$\Ideal{Sol( \Cc)}$. We use $\IMP_d(\Gamma)$ to denote $\IMP(\Gamma)$ when the input polynomial $f$ has degree at most $d$.
\end{definition}
The \GB basis $G$ of an ideal is a set of generators such that $f\in \GIdeal{G} \iff \reduce{f}{G}=0$, where $\reduce{f}{G}$ denotes the remainder of $f$ divided by $G$  (see~\cite{Cox} or \cref{sect:GBbasics} for more details and notations).
\begin{definition}\label{def:dTruncated GB}
  If $G$ is a \GB basis of an ideal, the \textbf{d-truncated \GB basis} $G'$ of $G$ is defined as
  $$G' = G \cap \mathbb{F}[x_1,x_1,\dots,x_n]_d,$$ where $\mathbb{F}[x_1,x_1,\dots,x_n]_d$ is the set of polynomials of degree less than or equal to $d$.
\end{definition}
It is not necessary to compute  a \GB basis of $\Ideal{Sol(\Cc)}$ in its entirety to solve the $\IMP_d$. Since the input polynomial $f$ has degree $d=O(1)$, the only polynomials from $G$ that can possibly divide $f$, in the graded lexicographic order (see \cref{def:lex and grlex} in \cref{sect:GBbasics}), are those that are in $G'$. The remainders of such divisions are also in $\mathbb{F}[x_1,x_1,\dots,x_n]_d$. Therefore, by \cref{th:gbprop,th:imp}, the membership test can be computed by
using only polynomials from $G'$ and therefore we have
$$f \in \Ideal{Sol(\Cc)} \cap \mathbb{F}[x_1,x_1,\dots,x_n]_d \iff \reduce{f}{G'}=0.$$
From the previous observations it follows that if we can compute $G'$ in $n^{O(d)}$ then this yields an algorithm that runs in $n^{O(d)}$ time for the $\IMP_d$ (note that the size of the input polynomial $f$ is bounded by $n^{
O(d)}$).

\subsection{Our contributions}\label{sec:Our Contributions}

In this paper we focus on instances $\mathcal{C}=(X=\{x_1,\dots,x_n\},D=\{0,1\},C)$ of CSP($\Gamma$) (see \cref{def:csp}) where $\Gamma$ is a language  that is closed under \Minorityns\  (see \cref{def:minority}). We first produce the reduced \GB basis $G_1$ of $\vb{I}(Sol(\Cc))$ according to the lexicographic order. Note that this \GB basis does not guarantee finding a solution to the $\IMP_d(\Gamma)$ in polynomial time.  
In \cref{sec:conversion} we show how to convert $G_1$ to a $d$-truncated \GB basis $G_2$ for a graded lexicographic monomial ordering. We prove that this computation can be obtained in polynomial time for any fixed $d=O(1)$. As pointed out at the end of \cref{sec:preliminaries}, an efficient computation of $G_2$ yields an efficient algorithm for the $\IMP_d$. A simple example is provided in \cref{sec:example}.
Thus we have the following main results:

\begin{theorem}\label{thm:MainTheorem}
The $d$-truncated reduced \GB basis of a Boolean combinatorial ideal whose constraints are closed under the minority polymorphism can be computed in $n^{O(d)}$ time, assuming the graded lexicographic ordering of monomials.
\end{theorem}
This proves the following:
\begin{corollary}\label{cor:MainCorollary}
  The $\IMP_d(\Gamma)$, over the Boolean domain, can be solved in polynomial time for $d=O(1)$ if the solution space of every constraint in $\Gamma$ is closed under the minority polymorphism.
\end{corollary}

\emph{Structure of the proof:} A high level description of the proof structure is as follows. Each constraint that is closed under the minority polymorphism can be written in terms of linear equations (mod 2) (see e.g. \cite{Chen09}). In \cref{sec:GBlex}, we first express these equations in their reduced row echelon form: that is to say the `leading variable' (the variable that comes first in the lexicographic order or \lex in short, see \cref{def:lex and grlex}) in each equation does not appear in any other (mod 2) equation.
We then show how each polynomial in (mod 2) translates to a polynomial in regular arithmetic with exactly the same 0/1 solutions. The use of elementary symmetric polynomials allows for an efficient computation of the polynomials in regular arithmetic. Using these, we produce a set of polynomials $G_1$ and prove that $G_1$ is the reduced \GB basis of $\vb{I}(Sol(\Cc))$ in the \lex order.
As already mentioned, a \GB basis in the \lex order does not guarantee the efficient solvability of the $\IMP_d$.
We provide a conversion algorithm in \cref{sec:conversion} which converts $G_1$ to the $d$-truncated reduced \GB basis $G_2$ of $\vb{I}(Sol(\Cc))$ in the graded lexicographic ordering (\grlex for short, see \cref{def:lex and grlex}). 
In \cref{sec:ProductOfBooleanFunctions} we show how polynomials in $G_1$ from \cref{sec:GBlex} are handled so our conversion algorithm in \cref{sec:OurAlgorithm} works in polynomial time. \cref{thm:correctness_algorithm} proves the correctness and polynomial running time of the conversion algorithm. This gives the proof of the main results of the paper stated in  \cref{thm:MainTheorem} and \cref{cor:MainCorollary}.

\section{\GB bases in \lex order}\label{sec:GBlex}

Consider an instance $\mathcal{C}=(X=\{x_1,\dots,x_n\},D=\{0,1\},C)$ of CSP($\Gamma$) where $\Gamma$ is a language  that is closed under \Minorityns. Any constraint of $\Cc$ can be written as a system of linear equations over $\textrm{GF}(2)$ (see e.g. \cite{Chen09}). 
These linear systems with variables $x_1,\dots,x_n$ can be solved by Gaussian elimination. If there is no solution, then we have from Hilbert's Weak Nullstellensatz (\cref{th:nullstz}) that $1\in \vb{I}(Sol(\Cc)) \iff Sol(\Cc)=\emptyset \iff \vb{I}(Sol(\Cc)) = \mathbb{R}[\vb{x}]$. If $1\in \vb{I}(Sol(\Cc))$ the reduced \GB basis is $\{1\}$. We proceed only if $Sol(\Cc)\neq\emptyset$.
In this section, we assume the \lex order $>_\lex$ with $x_1>_\lex x_2>_\lex \dots >_\lex x_n$. We also assume that the linear system has $r\leq n$ equations and is already in its reduced row echelon form with $x_i$ as the leading monomial of the $i$-th equation. Let $Supp_i\subset [n]$ such that $\{x_j:j\in Supp_i\}$ is the set of variables appearing in the $i$-th equation of the linear system except for $x_i$. Let the $i$-th equation be $R_i = 0 \Mod{2}$ where 
\begin{equation}\label{eq:Ri}
    R_i := x_i\oplus f_i,     
\end{equation}
with $i\in[r]$ and $f_i$ is the Boolean function $(\bigoplus_{j\in Supp_i} x_j)\oplus \alpha_i$ and $\alpha_i = 0/1$.

\subsection{From (mod 2) to regular arithmetic \GB basis} 
In this section, we show how to transform $R_i$'s into polynomials in regular arithmetic. 
The idea is to map $R_i$ to a polynomial $R_i'$ over $\Real[x_1, \ldots, x_n]$ such that $a\in \{0,1\}^n$ satisfies $R_i=0$ if and only if $a$ satisfies $R_i' = 0$. Moreover, $R_i$ is such that it has the same leading term as $R_i'$. 
We produce a set of polynomials $G_1$ and prove that $G_1$ is the reduced \GB basis of $\vb{I}(Sol(\Cc))$ over $\Real[x_1, \ldots, x_n]$ in the \lex ordering.
We define $R'_i$ as
\begin{align}
    R_{i}':= x_i - M(f_i) \label{eq:R'_i}
\end{align}
where
\begin{align}\label{eq:mod2regexpansion}
    M(f_i) &= 
    \begin{cases}
    \sum\limits_{k=1}^{|Supp_i|} \left( (-1)^{k-1}\cdot 2^{k-1}\sum\limits_{\{x_{j_1},\dots,x_{j_k}\}\subseteq Supp_i}x_{j_1}x_{j_2}\cdots x_{j_k}\right) \textrm{ when } \alpha_i=0\\
    1+\sum\limits_{k=1}^{|Supp_i|} \left( (-1)^{k}\cdot 2^{k-1}\sum\limits_{\{x_{j_1},\dots,x_{j_k}\}\subseteq Supp_i}x_{j_1}x_{j_2}\cdots x_{j_k}\right)  \textrm{ when } \alpha_i=1
    \end{cases}
\end{align}

\begin{lemma} Consider the following set of polynomials:
\begin{align}\label{eq:GrobBasisLex}
    G_1=\{R_1',\ldots, R_r',x_{r+1}^2-x_{r+1},\ldots, x_n^2-x_n\},
\end{align}
where $R_i'$ is from \cref{eq:R'_i}. $G_1$ is the reduced \GB basis of $\vb{I}(Sol(\Cc))$ in the lexicographic order $x_1>_\lex x_2>_\lex \dots,>_\lex x_n$.
\end{lemma}
\begin{proof} For any two Boolean variables $x$ and $y$,
\begin{align}\label{eq:mod2arith}
    x\oplus y &= x+y- 2xy.
\end{align}
By repeatedly using \cref{eq:mod2arith} to obtain the equivalent expression for $f_i$, we see that $R_i=0 \Mod{2}$ and $R_i'=0$ have the same set of 0/1 solutions. Therefore $\Variety{\GIdeal{G_1}}$ is equal to $Sol(\Cc)$. This implies that $\GIdeal{G_1} \subseteq \vb{I}(Sol(\Cc))$. Moreover, $\LM(R_i)=\LM(R_i')=x_i$, by construction. For every pair of polynomials in $G_1$ the reduced $S$-polynomial is zero as the leading monomials of any two polynomials in $G_1$ are relatively prime. By Buchberger's Criterion (see \cref{th:crit}) it follows that $G_1$ is a \GB basis of $\GIdeal{G_1}$ over $\Real[x_1,\ldots,x_n]$ (according to the \lex order). In fact, it can be seen by inspection that $G_1$ is the \textit{reduced} \GB basis of $\GIdeal{G_1}$.
To prove that $\vb{I}(Sol(\Cc)) = \GIdeal{G_1}$, we need to prove that any $p \in \vb{I}(Sol(\Cc)) \implies p\in \GIdeal{G_1}$. It is enough to prove that $p|_{G_1} = 0$ as this implies $p\in \GIdeal{G_1}$. We have that $p|_{G_1}$ cannot contain variable $x_i$ for all $1\leq i \leq r$. Hence $p|_{G_1}$ is multilinear in $x_{r+1},x_{r+2},\dots,x_n$.  Each tuple of $D^{n-r}$ extends to exactly that $n-$tuple in $Sol(\Cc)$ whose coordinate associated with $x_i$ ($1\leq i\leq r$) is the unique value $x_i$ takes to satisfy $x_i\oplus f_i = 0$ (see \cref{eq:Ri} and \cref{eq:R'_i}). As $p|_{G_1}$ is multilinear in $x_{r+1},x_{r+2},\dots,x_n$, there are at most $2^{n-r}$ coefficients. Since every point of $D^{n-r}$ is a solution of $p|_{G_1}$, we see that every coefficeint of $p|_{G_1}$ is zero and hence $\reduce{p}{G_1}$ is the zero polynomial. Hence $G_1$ is the reduced \GB basis of $\vb{I}(Sol(\Cc))$.
\end{proof}

\begin{example}
Consider a system with just one equation with $R_1:= x_1\oplus x_2 \oplus x_3 = 0$ where  $x_1>_\lex x_2>_\lex x_3$. Then $f_1:= x_2 \oplus x_3$ and $M(f_1):= x_2 + x_3 -2x_2x_3$. The polynomial corresponding to \cref{eq:R'_i} is
\begin{align*}
    R_{1}':= x_1-x_2-x_3+2x_2x_3.
\end{align*}
The equations $R_1=0$ and $R_{1}'=0$ have the same set of 0/1 solutions and  $\LM(R_1)=\LM(R_{1}')=x_1$. For every pair of polynomials in $G=\{R_{1}',x_2^2-x_2,x_3^2-x_3\}$ the reduced $S$-polynomial is zero. By Buchberger's Criterion (see e.g. \cite{Cox} or \cref{th:crit} in the appendix) it follows that $G$ is a \GB basis over $\Real[x_1,x_2,x_3]$ (according to the specified \lex order). 
\end{example}
%

Note that the reduced \GB basis in \cref{eq:GrobBasisLex} can be ``efficiently'' computed by exploiting the high degree of symmetry in each $M(f_i)$ and using a version of the elementary symmetric polynomials.


\section{Conversion of basis}\label{sec:conversion}

Now that we have the reduced \GB basis in \lex order, we show how to obtain the $d$-truncated reduced \GB basis in \grlex order in polynomial time for any fixed $d=O(1)$. Before we describe our conversion algorithm, we show how to expand a product of Boolean functions. This expansion will play a crucial step in our algorithm.    
\subsection{Expansion of a product of Boolean functions}\label{sec:ProductOfBooleanFunctions}
In this section, we show a relation between a product of Boolean functions and (mod 2) sums of the Boolean functions, which is heavily used in our conversion algorithm in \cref{sec:OurAlgorithm}. We have already seen from \cref{eq:mod2arith} that if $f,g$ are two Boolean functions,\footnote{We earlier considered Boolean variables, but the same holds for Boolean functions.}
then 
$$2\cdot f\cdot g = f+g-(f\oplus g).$$

Hence it can be proved by repeated use of the above equation that the following holds for Boolean functions $f_1,$ $f_2,\dots,f_m$:

\begin{equation}\label{eq:mod2expansion}
\begin{aligned}
f_1\cdot f_2 \cdots f_m = \frac{1}{2^{m-1}}\biggl[ &\sum_{i\in[m]} f_i - \sum_{\{i,j\}\subset [m]} (f_i\oplus f_j) + \sum_{\{i,j,k\}\subset [m]} (f_i\oplus f_j \oplus f_k) + \dots +\\ 
&(-1)^{m-1}(f_1\oplus f_2 \oplus \dots \oplus f_m)\biggr].
\end{aligned}
\end{equation}
We call each Boolean function of the form $(f_{i_1}\oplus\cdots\oplus f_{i_k})$ in \cref{eq:mod2expansion} as a \textbf{Boolean term}. We call the Boolean term $(f_1\oplus f_2 \oplus \dots \oplus f_m)$ as the \textbf{longest Boolean term} of the expansion.
Thus, a product of Boolean functions can be expressed as a linear combination of Boolean terms. Note that \cref{eq:mod2expansion} is \textit{symmetric} with respect to $f_1,f_2,\dots,f_m$ as any $f_i$ interchanged with $f_j$ produces the same expression. It is no coincidence that we chose the letter $f$ in the above equation: we later apply this identity using $f_j$ from $R_j:=x_j\oplus f_j$ (see \cref{sec:GBlex}). When we use \cref{eq:mod2expansion} in the conversion algorithm, we will have to evaluate a product of at most $d$ functions, i.e. $m\leq d=O(1)$.
We now see in the right hand side of \cref{eq:mod2expansion} that the coefficient $1/2^{m-1}$ is of constant size and there are $O(1)$ many Boolean terms.

\subsection{Our conversion algorithm}\label{sec:OurAlgorithm}
The FGLM \cite{FAUGERE1993329} conversion algorithm is well known in computer algebra for converting a given reduced \GB basis of a zero dimensional ideal in some ordering to the reduced \GB basis in any other ordering. However, it does so with $O(nD(\langle G_1 \rangle)^3)$ many arithmetic operations, where $D(\GIdeal{G_1})$ is the dimension of the $\mathbb{R}$-vector space $\mathbb{R}[x_1,\dots,x_n]/\GIdeal{G_1}$ (see Proposition 4.1 in \cite{FAUGERE1993329}). $D(\GIdeal{G_1})$ is also equal to the number of common zeros (with multiplicity) of the polynomials from $\GIdeal{G_1}$, which would imply that for the combinatorial ideals considered in this paper, $D(\GIdeal{G_1})=O(2^{n-r})$. This exponential running time is avoided in our conversion algorithm by exploiting the symmetries in \cref{eq:mod2regexpansion} and by truncating the computation up to degree $d$. 



Some notations necessary for the algorithm are as follows: $G_1$ and $G_2$ are the reduced \GB basis of $\GIdeal{G_1}$ in \lex and \grlex ordering respectively. $\LM(G_i)$ is the set of leading monomials of polynomials in $G_i$ for $i\in\{1,2\}$. Since we know $G_1$, we know $\LM(G_1)$, whereas $G_2$ and $\LM(G_2)$ are constructed by the algorithm. $B(G_1)$ is the  set of monomials that cannot be divided (considering the \lex order) by any monomial of $\LM(G_1)$. Therefore, $B(G_1)$ is the set of all multilinear monomials in variables $x_{r+1},\dots,x_n$. Similarly, $B(G_2)$ is the set of monomials that cannot by divided (considering the \grlex order) by any monomial of $\LM(G_2)$. 

Recall the definition of $f_i$ for $i\leq r$ from \cref{sec:GBlex}. For $i>r$, for notational purposes, we define the Boolean function $f_i:=x_i$.

\begin{lemma}\label{lem:q in Boolean Terms}
Consider a monomial $q$ such that $deg(q)\leq d$. Then $\reduce{q}{G_1}$ can be expressed as a linear combination of Boolean terms.
\end{lemma}
\begin{proof}
Consider $q=x_{i_1}x_{i_2}\cdots x_{i_k}$ where $k\leq d$. Then from \cref{eq:Ri,eq:R'_i}, $\reduce{q}{G_1} = f_{i_1}f_{i_2}\cdots f_{i_k}$ and the lemma holds using \cref{eq:mod2expansion}.
\end{proof}
Let elements $b_i$ of $B(G_2)$ be arranged in increasing \grlex order. We construct a set $C$ in our algorithm such that its elements $c_i$ are defined as $c_i=\reduce{b_i}{G_1}$ written as linear combinations of Boolean terms using \cref{lem:q in Boolean Terms}. 
We say that a Boolean term $f$ of $c_i$ ``appears in $c_j$'' for some $j<i$ if the longest Boolean term of $c_j$ is $f\oplus\alpha$ where $\alpha=0/1$.\\ 

Let $Q$ be the set of all monomials $m$ such that $1<_\grlex deg(m) \leq_\grlex d$. We recommend the reader to refer to the example in \cref{sec:example} and \cref{sec:Example in detail} for an intuitive working of the algorithm. We now describe the algorithm in full (we assume $1 \notin \vb{I}(Sol(\Cc))$, else $G_1=\{1\}=G_2$ and we are done): \\

\noindent \textbf{Inputs:} Degree $d$, $G_1$, $Q$ \\
\textbf{Initial states:} $G_2= \emptyset$, $B(G_2)=\{1(=b_1)\}$, $C=\{1(=c_1)\}$, $q=x_n$.\\
\textbf{Outputs:} $d$-Truncated versions of $G_2$, $B(G_2)$.

\begin{itemize}
    \item \textbf{Main loop:}
    Find ${q}|_{G_1}$, by which we simply replace any occurrence of $x_i$ by the Boolean functions $f_i$. Expand ${q}|_{G_1}$ by using \cref{eq:mod2expansion}. 
    \begin{itemize}
        \item Suppose the longest Boolean term of ${q}|_{G_1}$ does not appear in any $c\in C$. Then ${q}|_{G_1}$ is written as a linear combination of $\reduce{b_i}{G_1}$ and its longest Boolean term (see \cref{lem:C}). This polynomial is added to $C$ and $q$ is added to $B(G_2)$. Go to \textbf{Termination check}.
        \item If the longest Boolean term of ${q}|_{G_1}$ appears in some $c\in C$, then every Boolean term of ${q}|_{G_1}$ can be written as linear combinations of $\reduce{b_j}{G_1}$'s. Note that if the longest Boolean term $f$ appears in $c$ as $f\oplus 1$, then we use $f\oplus 1 = 1- (f)$ (see \cref{eq:mod2arith}). Thus we have ${q}|_{G_1}=\sum_j k_j\reduce{b_j}{G1} \implies q-\sum_j k_jb_j \in \GIdeal{G_1}$. The polynomial $q-\sum_j k_jb_j$ is added to $G_2$ and $q$ to $\LM(G_2)$. Go to \textbf{Termination check}.
    \end{itemize}
    \item \textbf{Termination check:} We delete the occurrence of $q$ from $Q$. If $q$ was added to $\LM(G_2)$ then we delete any monomial in $Q$ that $q$ can divide. The algorithm terminates if $Q$ is empty, else go to \textbf{Next monomial}.
    \item \textbf{Next monomial:} Choose the smallest (according to \grlex order) monomial in $Q$ as $q$. Go to \textbf{Main loop}.
\end{itemize}

\begin{lemma}\label{lem:C}
The set $C$ is such that every $c_i$ is a linear combination of existing $\reduce{b_j}{G_1}$'s ($j<i$) and the longest Boolean term of $\reduce{b_i}{G_1}$.
\end{lemma}

\begin{proof}
By definition, element $c_i$ is added to $C$ when a monomial $q$ is added to $B(G_2)$ where $b_i=q$ and $c_i=\reduce{b_i}{G_1}$ expressed in Boolean terms (see Main loop). This means that $q$ is not divisible by any monomial in $\LM(G_2)$. We prove the lemma by induction on the degree of $q$. Note that $b_1 = 1$ and hence $c_1 = \reduce{b_1}{G_1} = 1$.

If $deg(q)=1$, then $q$ is some $x_i$ and $\reduce{x_i}{G_1}$ is one of $0,1$ or $f_i$. If $\reduce{x_i}{G_1}$ is either 0 or 1, then it then appears in $c_1$. We are now in the second case of the Main loop, so $q$ should be added to $\LM(G_2)$ and not $B(G_2)$. Hence $\reduce{x_i}{G_1}$ can be neither 0 nor 1 and the lemma holds for $deg(q)=1$ as $f_i$ is the longest Boolean term.

Let us assume the statement holds true for all monomials with degree less than $m$. Consider $q$ such that $deg(q)=m$ and $q=x_{i_1}x_{i_2}\dots x_{i_m}$ where $i_j$'s need not be distinct, and the lemma holds for every monomial $<_{\grlex} q$. Then $\reduce{q}{G_1} = f_{i_1}\cdot f_{i_2}\cdots f_{i_m}$.
Let $(f_{j_1}\oplus \cdots \oplus f_{j_k})$ be a Boolean term in the expansion of $\reduce{q}{G_1}$ (by using \cref{eq:mod2expansion}), that is not the longest Boolean term, so $\{j_1,\dots,j_k\}\subset \{i_1,\dots,i_m\}$ and $k<m$. 
Consider the monomial $x_{j_1}x_{j_2}\dots x_{j_k}$. We will now prove that $x_{j_1}x_{j_2}\dots x_{j_k}$ is in fact some $b_l\in B(G_2)$ and there exists $c_l\in C$ which is a linear combination of $\reduce{b_i}{G_1}$'s and $(f_{j_1}\oplus \cdots \oplus f_{j_k})$.
The monomial $x_{j_1}x_{j_2}\dots x_{j_k}$ either belongs to $LM(G_2)$ or $B(G_2)$.
If $x_{j_1}x_{j_2}\dots x_{j_k}\in\LM(G_2)$ then it divides $q$, a contradiction to our choice of $q$. Therefore, $x_{j_1}x_{j_2}\dots x_{j_k}=b_l\in B(G_2)$.
Clearly $b_l<_{\grlex}q$ and the induction hypothesis applies, so there exists $c_l\in C$ such that
$$\reduce{b_l}{G_1}=c_l=\sum_{i < l} a_i\reduce{b_i}{G_1} + a_0(f_{j_1}\oplus \cdots \oplus f_{j_k})$$
where $a_i$'s are constants.
Then we simply use the above equation to substitute for the Boolean term $f_{j_1}\oplus \cdots \oplus f_{j_k}$ in $\reduce{q}{G_1}$ as a linear combination of $\reduce{b_i}{G_1}$ where $i\leq l$. We can do this for every Boolean term of $\reduce{q}{G_1}$ except the longest one. Hence the lemma holds.
\end{proof}

\begin{theorem}\label{thm:correctness_algorithm}
The conversion algorithm terminates for every input $G_1$ and correctly computes a $d$-truncated reduced \GB basis, with the \grlex ordering, of the ideal $\GIdeal{G_1}$ in polynomial time.
\end{theorem}
\begin{proof}
The Main loop runs at most $|Q| = O(n^d)$ times. Evaluation of any $\reduce{q}{G_1}$ can be done in $O(n)$ steps (see \cref{eq:mod2expansion}), checking if previous $c_i$'s appear (and replacing every Boolean term appropriately if it does) takes at most $O(n^d)$ steps since there are at most $|Q|$ many elements in $C$. Hence the running time of the algorithm is $O(n^{2d})$.

Suppose the set of polynomials $\{g_1, g_2, \dots, g_k\}$ is the output of the algorithm for some input $G_1$. Clearly, $deg(g_i)\leq d$ for all $i\in [k]$. We now prove by contradiction that the output is the $d$-truncated \GB basis of the ideal $\GIdeal{G_1}$ with the \grlex ordering.  Suppose $g$ is a polynomial of the ideal with $deg(g)\leq d$, but no $\LM(g_i)$ can divide $\LM(g)$. In fact, since every $g_i\in \GIdeal{G_1}$ we can replace $g$ by $\reduce{g}{\{g_1, g_2, \dots, g_k\}}$ ($g$ generalises the reduced $S$-polynomial). The fact that $g\in\GIdeal{G_1}$ and $\reduce{g}{G_1}=0$ implies that $\LM(g)$ is a linear combination of monomials that are less than $\LM(g)$ (in the \grlex order) and hence must be in $B(G_2)$, i.e 
$$\reduce{g}{G_1}=0 \implies \reduce{\LM(g)}{G_1}=\sum_i k_i\reduce{b_i}{G_1}$$
where every $b_i\in B(G_2)$ and $b_i<_\grlex \LM(g)$. When the algorithm runs for $q=\LM(g)$, since $q$ was not added to $\LM(G_2)$,
$$\reduce{\LM(g)}{G_1}= \sum_j k_j\reduce{b_j}{G_1} + f$$
where $f$ is the longest Boolean term of $\reduce{\LM(g)}{G_1}$ which does not appear in any previous element of $C$. But the two equations above imply that $\sum_i k_i\reduce{b_i}{G_1} = \sum_j k_j\reduce{b_j}{G_1} + f$, which proves that there exists some $b_l\in B(G_2)$ such that $c_l$ has $f$ as its longest Boolean term, so $f$ should have appeared in $c_l$, a contradiction. Therefore the output is a $d$-truncated \GB basis. Although unnecessary for the $\IMP_d$, we also prove that the output is reduced: every non leading monomial of every polynomial in the output comes from $B(G_2)$ and no leading monomial is a multiple of another by construction (see Termination check).
\end{proof}
Thus we have proof of the main theorem and corollary (see \cref{thm:MainTheorem} and \cref{cor:MainCorollary}).

\section{An example}\label{sec:example}
We provide a simple example in \cref{tab:example} where we convert the reduced \GB basis in \lex order of a combinatorial ideal to one in \grlex order. Consider the problem formulated by the following (mod 2) equations: $x_1 \oplus x_3 \oplus x_4 = 0$ and $x_2 \oplus x_3 \oplus x_5 \oplus 1 = 0$. The example is explained in more detail in \cref{sec:Example in detail}.

\begin{table}[ht]
    \centering
    \begin{tabular}{c|c|c|c|c}
    \#  &  $q$ & $B(G_2)$ & $C$ & $G_2$\\
    \hline
    0 &  -  & 1 & $1$ & $\emptyset$ \\
    1 & $x_5$ & $x_5$ & $x_5$ & - \\
    2 & $x_4$ & $x_4$ & $x_4$ & -\\
    3\label{seehere} & $x_3$ & $x_3$ & $x_3$ & -\\
    4 & $x_2$ & $x_2$ & $x_3 \oplus x_5 \oplus 1$ & -\\
    5 & $x_1$ & $x_1$ & $x_3 \oplus x_4 $ & -\\
    6 & $x_5^2$ & - & - & $x_5^2-x_5$\\
    7 & $x_4x_5$ & $x_4x_5$ & $\frac{1}{2}[\reduce{x_4}{G_1}+\reduce{x_5}{G_1}$ & -\\
     &  &  & $-(x_4 \oplus x_5)]$ & \\
    8 & $x_4^2$ & - & - & $x_4^2-x_4$\\
    9 & $x_3x_5$ & - & - & $x_3x_5-\frac{1}{2}[x_2+x_3+x_5-1]$\\
    10 & $x_3x_4$ & - & - & $x_3x_4 - \frac{1}{2}[-x_1+x_3+x_4]$\\
    11 & $x_3^2$ & - & - & $x_3^2-x_3$\\
    12 & $x_2x_5$ & -  & - & $x_2x_5 - \frac{1}{2}[x_2 + x_3 + x_5 - 1]$\\
    13 & $x_2x_4$ & $x_2x_4$ & $\frac{1}{2}[{x_2}|_{G_1} + {x_4}|_{G_1}$ & -\\
     &  &  & $ - (x_3\oplus x_4 \oplus x_5 \oplus 1) ]$ & \\
    14 & $x_2x_3$ & - & - & $x_2x_3-\frac{1}{2}[x_2+x_3+x_5-1]$\\
    15 & $x_2^2$ & - & - & $x_2^2-x_2$\\
    16 & $x_1x_5$ & - & - & $x_1x_5+x_2x_4-\frac{1}{2}[x_1+x_2+x_4+x_5-1]$\\
    17 & $x_1x_4$ & - & - & $x_1x_4-\frac{1}{2}[x_1-x_3+x_4]$\\
    18 & $x_1x_3$ & - & - & $x_1x_3-\frac{1}{2}[x_1+x_3-x_4]$\\
    19 & $x_1x_2$ & - & - & -\\
    20 & $x_1^2$ & - & - & $x_1^2-x_1$\\

\end{tabular}
    \caption{Example}
    \label{tab:example}
\end{table}

\section{Conclusion}
The $\IMP_d$ tractability for combinatorial ideals has useful practical applications as it implies bounded coefficients in Sum-of-Squares proofs.
A dichotomy result between ``hard'' (NP-hard) and ``easy'' (polynomial time) $\IMP$s was recently achieved for the $\IMP_0$ ~\cite{Bulatov17,Zhuk17} over the finite domain nearly thirty years after that over the Boolean domain \cite{Schaefer78}. The $\IMP_d$ for $d=O(1)$ over the Boolean domain was tackled by Mastrolilli \cite{MonaldoMastrolilli2019} based on the classification of the $\IMP$ through polymorphisms, where the complexity of the $\IMP_d$ for five of six polymorphisms was solved. We solve the remaining problem, i.e. the complexity of the $\IMP_d(\Gamma)$ when $\Gamma$ is closed under the ternary minority polymorphism. This is achieved by showing that the $d$-truncated reduced \GB basis can be computed in polynomial time, thus completing the missing link in the dichotomy result of \cite{MonaldoMastrolilli2019}. 

Moreover, we believe the techniques described in this paper can be generalized for a finite domain with prime $p$ elements, as constraints that are linear equations (mod $p$) are associated with an affine polymorphism \cite{Jeavons_UnifyingFramework1995}. We claim that the $\IMP_d$ is tractable for problems that are constrained as linear equations (mod $p$). This is a step in identifying the borderline of tractability, if it exists, for the general $\IMP_d$.
We believe that generalizing the dichotomy results of solvability of the IMP$_d$ for a finite domain is an interesting and challenging goal that we leave as an open problem.

{\small
\bibliographystyle{abbrv}
\bibliography{ref}

\begin{thebibliography}{10}

\bibitem{BUCHBERGER2006475}
B.~Buchberger.
\newblock Bruno buchberger’s phd thesis 1965: An algorithm for finding the
  basis elements of the residue class ring of a zero dimensional polynomial
  ideal.
\newblock {\em Journal of Symbolic Computation}, 41(3):475 -- 511, 2006.
\newblock Logic, Mathematics and Computer Science: Interactions in honor of
  Bruno Buchberger (60th birthday).

\bibitem{Bulatov17}
A.~A. Bulatov.
\newblock A dichotomy theorem for nonuniform {CSP}s (best paper award).
\newblock In {\em 58th {IEEE} Annual Symposium on Foundations of Computer
  Science, {FOCS} 2017, Berkeley, CA, USA, October 15-17, 2017}, pages
  319--330, 2017.

\bibitem{2017dfu7}
A.~A. Bulatov.
\newblock Constraint satisfaction problems: Complexity and algorithms.
\newblock {\em ACM SIGLOG News}, 5(4):4--24, Nov. 2018.

\bibitem{Chen09}
H.~Chen.
\newblock A rendezvous of logic, complexity, and algebra.
\newblock {\em ACM Comput. Surv.}, 42(1):2:1--2:32, Dec. 2009.

\bibitem{Cox}
D.~A. Cox, J.~Little, and D.~O'Shea.
\newblock {\em Ideals, Varieties, and Algorithms: An Introduction to
  Computational Algebraic Geometry and Commutative Algebra}.
\newblock Springer Publishing Company, Incorporated, 4th edition, 2015.

\bibitem{Dickenstein1991TheMP}
A.~Dickenstein, N.~Fitchas, M.~Giusti, and C.~Sessa.
\newblock The membership problem for unmixed polynomial ideals is solvable in
  single exponential time.
\newblock {\em Discrete Applied Mathematics}, 33:73--94, 1991.

\bibitem{FAUGERE1993329}
J.-C. Faug{\`e}re, P.~M. Gianni, D.~Lazard, and T.~Mora.
\newblock Efficient computation of zero-dimensional gröbner bases by change of
  ordering.
\newblock {\em Journal of Symbolic Computation}, 16(4):329 -- 344, 1993.

\bibitem{HilbertBasisTheorem}
D.~Hilbert.
\newblock Ueber die theorie der algebraischen formen.
\newblock {\em Mathematische Annalen}, 36:473 -- 534, 1890.

\bibitem{Hilbert1893}
D.~Hilbert.
\newblock Ueber die vollen invariantensysteme.
\newblock {\em Mathematische Annalen}, 42:313--373, 1893.

\bibitem{JEAVONS1998185}
P.~Jeavons.
\newblock On the algebraic structure of combinatorial problems.
\newblock {\em Theoretical Computer Science}, 200(1):185 -- 204, 1998.

\bibitem{Jeavons_UnifyingFramework1995}
P.~Jeavons, D.~Cohen, and M.~Gyssens.
\newblock A unifying framework for tractable constraints.
\newblock In U.~Montanari and F.~Rossi, editors, {\em Principles and Practice
  of Constraint Programming --- CP '95}, pages 276--291, Berlin, Heidelberg,
  1995. Springer Berlin Heidelberg.

\bibitem{Laurent2009}
M.~Laurent.
\newblock {\em Sums of Squares, Moment Matrices and Optimization Over
  Polynomials}, pages 157--270.
\newblock Springer New York, New York, NY, 2009.

\bibitem{MonaldoMastrolilli2019}
M.~Mastrolilli.
\newblock The complexity of the ideal membership problem for constrained
  problems over the boolean domain.
\newblock In {\em Proceedings of the Thirtieth Annual ACM-SIAM Symposium on
  Discrete Algorithms}, SODA '19, pages 456--475, Philadelphia, PA, USA, 2019.
  Society for Industrial and Applied Mathematics.

\bibitem{Mayr1989}
E.~W. Mayr.
\newblock Membership in polynomial ideals over q is exponential space complete.
\newblock In B.~Monien and R.~Cori, editors, {\em STACS 89}, pages 400--406,
  Berlin, Heidelberg, 1989. Springer Berlin Heidelberg.

\bibitem{MAYR1982305}
E.~W. Mayr and A.~R. Meyer.
\newblock The complexity of the word problems for commutative semigroups and
  polynomial ideals.
\newblock {\em Advances in Mathematics}, 46(3):305--329, 1982.

\bibitem{odonnell2017}
R.~O'Donnell.
\newblock {SOS Is Not Obviously Automatizable, Even Approximately}.
\newblock In C.~H. Papadimitriou, editor, {\em 8th Innovations in Theoretical
  Computer Science Conference (ITCS 2017)}, volume~67 of {\em Leibniz
  International Proceedings in Informatics (LIPIcs)}, pages 59:1--59:10,
  Dagstuhl, Germany, 2017. Schloss Dagstuhl--Leibniz-Zentrum fuer Informatik.

\bibitem{raghavendra_weitz2017}
P.~Raghavendra and B.~Weitz.
\newblock {On the Bit Complexity of Sum-of-Squares Proofs}.
\newblock In I.~Chatzigiannakis, P.~Indyk, F.~Kuhn, and A.~Muscholl, editors,
  {\em 44th International Colloquium on Automata, Languages, and Programming
  (ICALP 2017)}, volume~80 of {\em Leibniz International Proceedings in
  Informatics (LIPIcs)}, pages 80:1--80:13, Dagstuhl, Germany, 2017. Schloss
  Dagstuhl--Leibniz-Zentrum fuer Informatik.

\bibitem{Schaefer78}
T.~J. Schaefer.
\newblock The complexity of satisfiability problems.
\newblock In {\em Proceedings of the Tenth Annual ACM Symposium on Theory of
  Computing}, STOC '78, pages 216--226, New York, NY, USA, 1978. ACM.

\bibitem{PIT}
A.~Shpilka.
\newblock Recent results on polynomial identity testing.
\newblock In A.~Kulikov and N.~Vereshchagin, editors, {\em Computer Science --
  Theory and Applications}, pages 397--400, Berlin, Heidelberg, 2011. Springer
  Berlin Heidelberg.

\bibitem{vandongenPhd}
M.~R. van Dongen.
\newblock {\em Constraints, Varieties, and Algorithms}.
\newblock PhD thesis, Department of Computer Science, University College, Cork,
  Ireland, 2002.

\bibitem{Zhuk17}
D.~Zhuk.
\newblock A proof of {CSP} dichotomy conjecture (best paper award).
\newblock In {\em 58th {IEEE} Annual Symposium on Foundations of Computer
  Science, {FOCS} 2017, Berkeley, CA, USA, October 15-17, 2017}, pages
  331--342, 2017.

\end{thebibliography}
}
\newpage
\appendix

\section{Example (in detail)}\label{sec:Example in detail}
Note that $f_1 = x_3 \oplus x_4$, $f_2 = x_3 \oplus x_5 \oplus 1$, $f_3 = x_3$, $f_4 = x_4$ and $f_5 = x_5$. The reduced \GB basis in the \lex order is $G=G_1=\{x_1-M(f_1), x_2 - M(f_2)$, $dom(x_3)$, $dom(x_4)$, $dom(x_5)\}$. We start with $G_2=\LM(G_2)=\emptyset$, $B(G_2)=C=\{1\}$ (so $b_1=c_1=1$) and $q=x_5$. For the problem of $d=2$, we have
\begin{align*}
    Q = \{x_5,x_4,x_3,x_2,x_1,x_5^2, x_5x_4, x_4^2, x_5x_3,& x_4x_3, x_3^2, x_5x_2, x_4x_2, x_3x_2, x_2^2, x_5x_1, x_4x_1,\\ 
    &x_3x_1, x_2x_1, x_1^2\}.
\end{align*}

\phantomsection\label{ex:deg1}
We start with $q=x_5$ and since ${q}|_{G_1} = x_5$ and does not appear as the longest Boolean term of any element of $C$, we have that $x_5$ is added to $C$ (so $c_2=x_5$) and $x_5$ is added to $B(G_2)$ (so $b_2=x_5$).
The \textsf{Termination check} of the algorithm deletes $x_5$ from $Q$ and \textsf{Next monomial} chooses $q=x_4$.
The iterations are similar for $q=x_4$ and $q=x_3$, so we have $b_3=c_3=x_4$ and $b_4=c_4=x_3$ and $x_4,x_3$ are deleted from $Q$. When \textsf{Next monomial} chooses $q=x_2$, we have ${q}|_{G_1} = f_2 = (x_3\oplus x_5 \oplus 1)$, and since the Boolean term does not appear in any $c\in C$, we add $(x_3\oplus x_5 \oplus 1)$ to $C$ (so $c_5=(x_3\oplus x_5 \oplus 1)$) and $x_2$ to $B(G_2)$ (so $b_5=x_2$).
For similar reasons, when $q=x_1$, we add $c_6= (x_3\oplus x_4) $ to $C$ and $b_6=x_1$ to $B(G_2)$.

After the 5-th iteration (see \cref{tab:example}) is complete, we only have degree-two monomials in $Q$. \textsf{Next monomial} chooses $q=x_5^2$ and ${q}|_{G_1} = x_5$. Since $c_1=x_5$, $x_5$ appears as a Boolean term in $c_1$. Since the longest Boolean term appears already in $C$, ${q}|_{G_1}$ must be a linear combination of existing $\reduce{b_i}{G_1}$'s. That is to say, ${x_5^2}|_{G_1}=c_1 = \reduce{b_1}{G_1}={x_5}|_{G_1} \implies \reduce{x_5^2}{G_1}=\reduce{x_5}{G_1}$, so the polynomial $x_5^2-x_5$ is added to $G_2$. \textsf{Termination check} adds $x_5^2$ to $\LM(G_2)$ and deletes $x_5^2$ from $Q$. 

\textsf{Next monomial} chooses $q=x_5x_4$, so 
\begin{align*}
    {x_5x_4}|_{G_1}&= f_4\cdot f_5 =\frac{1}{2}[x_4+x_5-(x_4 \oplus x_5)]=\frac{1}{2}[\reduce{x_4}{G_1}+\reduce{x_5}{G_1}-(x_4 \oplus x_5)].
\end{align*}
The longest Boolean term of ${q}|_{G_1}$ is $(x_4\oplus x_5)$ which does not appear in any $c\in C$, so
$c_7 = 1/2[\reduce{x_4}{G_1}+\reduce{x_5}{G_1}-(x_4 \oplus x_5)]$ is added to $C$ and $b_7 = x_5x_4$ is added to $B(G_2)$.
\textsf{Next monomial} chooses $q=x_4^2$, this is similar to the case when $q=x_5^2$, we see that when $q=x_4^2$, and $x_4^2-x_4$ is added to $G_2$ and $x_4^2$ to $\LM(G_2)$. When \textsf{Next monomial} chooses $q=x_3x_5$ we have
\begin{align*}
    {x_3x_5}|_{G_1}&= f_3\cdot f_5=\frac{1}{2}[x_3+x_5-(x_3 \oplus x_5)].
\end{align*}
Note that $(x_3\oplus x_5 \oplus 1)$ appears in $c_5\in C$. We use the fact that $(f\oplus 1) = 1 - f$ (see \textsf{Main loop}), and we have
\begin{align*}
    {x_3x_5}|_{G_1}&= \frac{1}{2}[x_3+x_5-(x_3 \oplus x_5)] =\frac{1}{2}[\reduce{x_3}{G_1}+\reduce{x_5}{G_1}-(1 - (x_3 \oplus x_5\oplus 1))]\\
    &= \frac{1}{2}[\reduce{x_2}{G_1}+\reduce{x_3}{G_1}+\reduce{x_5}{G_1}-\reduce{1}{G_1}]
\end{align*}
and thus $x_3x_5-\frac{1}{2}[x_2+x_3+x_5-1]$ is added to $G_2$ and $x_3x_5$ to $\LM(G_2)$. The rest of the polynomials in $B(G_2),G_2,C$ are as shown in \cref{tab:example}. It can be seen that after the 20-th iteration, $Q$ becomes empty and \textsf{Termination check} halts the algorithm. This gives the 2-truncated reduced \GB basis $G_2$ of the combinatorial ideal. Note that this is in fact the reduced \GB basis in its entirety for this example (see \textsf{Termination check}).

\section{Ideals, Varieties and Constraints}\label{sect:background}
Let $\Field$ denote an arbitrary field (for the applications of this paper $\Field=\Real$). Let $\Field[x_1, \ldots, x_n]$ be the ring of polynomials over a field $\Field$ and indeterminates $x_1,\ldots, x_n$. Let $\Field[x_1, \ldots, x_n]_d$ denote the subspace of polynomials of degree at most $d$.
\begin{definition}\label{def:ideal}
The ideal (of $\Field[x_1,\ldots,x_n]$) generated by a finite set of polynomials $\{f_1,$ $\dots, f_m\}$ in $\Field[x_1,\ldots,x_n]$ is defined as
$$\Ideal{ f_1,\ldots,f_m}\mydef \left\{\sum_{i=1}^m t_i f_i\ |\ t_1,\ldots,t_m\in \Field[x_1,\ldots,x_n]\right\}.$$
The set of polynomials that vanish in a given set $S\subset \Field^n$ is called the \textbf{\emph{vanishing ideal}} of $S$ and denoted:
$\Ideal{S}\mydef \{f\in \Field[x_1,\ldots,x_n]: f(a_1,\ldots,a_n)=0 \  \forall (a_1,\ldots,a_n)\in S\}$.
\end{definition}
\begin{definition}
An ideal $\I$ is \textbf{\emph{radical}} if $f^m \in\I$ for some integer $m\geq 1$ implies that $f\in \I$.
\end{definition}
Another common way to denote $\Ideal{ f_1,\ldots,f_m}$ is by $\langle f_1,\ldots,f_m \rangle$ and we will use both notations interchangeably.
\begin{definition}
%
 Let $\{f_1,\ldots, f_m\}$ be a finite set of polynomials in $\Field[x_1,\ldots,x_n]$. We call
$\Variety{ f_1,\ldots,f_m}\mydef \{(a_1,\ldots,a_n)\in \Field^n|$ $ f_i(a_1,$ $\ldots,a_n)=0 \quad 1\leq i\leq m\}$
the \textbf{\emph{affine variety}} defined by $f_1,\ldots, f_m$.
\end{definition}
\begin{definition}\label{def:V(I)}
  Let $\I\subseteq \Field[x_1,\ldots,x_n]$ be an ideal. We will denote by $\Variety{\I}$ the set $\Variety{\I}=\{(a_1,\ldots,a_n)\in \Field^n| f(a_1,\ldots,a_n)=0 \quad \forall f\in \I\}$.
\end{definition}

\begin{theorem}[\cite{Cox}, Th.15, p.196]\label{th:ideal_intersection}
  If $I$ and $J$ are ideals in $\Field[x_1,\ldots,$ $x_n]$, then $\Variety{I\cap J}= \Variety{I}\cup \Variety{J}$.
\end{theorem}


 


\subsection{The Ideal-CSP Correspondence}\label{sect:idealCSP}
Indeed, let $\Cc=(X,D,C)$ be an instance of the $\CSP(\Gamma)$ (see ~\cref{def:csp}). Without loss of generality, we shall assume that $D\subset \N$ and $D\subseteq \Field$.

Let $Sol(\Cc)$ be the (possibly empty) set of all feasible solutions of $\Cc$.
In the following, we map $Sol(\Cc)$ to an ideal $\I_\Cc\subseteq \Field[X]$ such that $Sol(\Cc)=\Variety{\I_\Cc}$.

Let $Y=(x_{i_1},\ldots,x_{i_k})$ be a $k$-tuple of variables from $X$
and let $R(Y)$ be a non empty constraint from $C$.
In the following, we map $R(Y)$ to a generating system of an ideal such that the projection of the variety of this ideal onto $Y$ is equal to $R(Y)$ (see~\cite{vandongenPhd} for more details).

Every $v=(v_1,\ldots,v_k)\in R(Y)$ corresponds to some point $v\in \Field^k$. It is easy to check~\cite{Cox} that $\Ideal{\{v\}}= \GIdeal{x_{i_1}-v_{1},\ldots,x_{i_k}-v_{k}}$, where $\GIdeal{x_{i_1}-v_{1},\ldots,x_{i_k}-v_{k}}\subseteq \Field[Y]$ is radical.
By \cref{th:ideal_intersection}, we have
\begin{align}\label{eq:constr=var}
 R(Y)=\bigcup_{v\in R(Y)} \Variety{\Ideal{\{v\}}}=\Variety{\I_{R(Y)}}
 \qquad \text{where } \I_{R(Y)} = \bigcap_{v\in R(Y)} \Ideal{\{v\}},
\end{align}
where $\I_{R(Y)}\subseteq \Field[Y]$ is zero-dimensional and radical ideal since it is the intersection of radical ideals (see~\cite{Cox}, Proposition~16, p.197). Equation~\eqref{eq:constr=var} states that constraint $R(Y)$ is a variety of $\Field^k$. It is easy to find a generating system for $\I_{R(Y)}$:
 \begin{align}\label{eq:genIConstr}
\I_{R(Y)}=\langle\prod_{v\in R}(1-\prod_{j=1}^{k}\delta_{v_j}(x_{i_j})),\prod_{j\in D}(x_{i_1}-j),
   \ldots,\prod_{j\in D}(x_{i_k}-j)\rangle,
 \end{align}
where $\delta_{v_j}(x_{i_j})$ are indicator polynomials, i.e. equal to one when $x_{i_j}=v_j$ and zero when $x_{i_j}\in D\setminus\{v_j\}$; polynomials $\prod_{j\in D}(x_{i_k}-j)$ force variables to take values in $D$ and will be denoted as {\emph{domain polynomials}}.

The smallest ideal (with respect to inclusion) of $\Field[X]$ containing $\I_{R(Y)}\subseteq \Field[\x]$ will be denoted $\I_{R(Y)}^{\Field[X]}$ and it is called the $\Field[X]$-module of $\I$. The set $Sol(\Cc)\subset \Field^n$ of solutions of $\Cc=(X,D,C)$ is the intersection of the varieties of the constraints:
\begin{align}
  Sol(\Cc) & = \bigcap_{R(Y)\in C}\Variety{\I_{R(Y)}^{\Field[X]}}=\Variety{\I_C} \label{eq:solVar},\\
  \I_\Cc&=\sum_{R(Y)\in C}\I_{R(Y)}^{\Field[X]}.\label{eq:IC}
\end{align}

The following properties follow from Hilbert's Nullstellensatz. 
\begin{theorem}\label{th:nullstz}
Let $\Cc$ be an instance of the $\CSP(\Gamma)$ and $\I_\Cc$ defined as in~\eqref{eq:IC}. Then
  \begin{align}
    & \text{(Weak Nullstellensatz)} \label{eq:weak_nstz}\\ \nonumber
    &\Variety{\I_\Cc}=\emptyset \Leftrightarrow 1\in \Ideal{\I_\Cc} \Leftrightarrow \I_\Cc=\Field[X],  \\
    &\label{eq:strong_nstz} \text{(Strong Nullstellensatz)}\\ \nonumber
    &\Ideal{\Variety{\I_\Cc}}=\sqrt{\I_\Cc},\\
    &\label{eq:ICradical} \text{(Radical Ideal)} \\ \nonumber
    &\sqrt{\I_\Cc}=\I_\Cc.
  \end{align}
\end{theorem}
\cref{th:nullstz} follows from a simple application of the celebrated and basic result in algebraic geometry known as Hilbert's Nullstellensatz. In the general version of Nullstellensatz it is necessary to work in an algebraically closed field and take a radical of the ideal of polynomials. In our special case it is not needed due to the presence of domain polynomials. Indeed, the latter implies that we know a priori that the solutions must be in $\Field$ (note that we are assuming $D\subseteq \Field$).

\subsection{\GB bases.}\label{sect:GBbasics}
In this section we suppose a fixed monomial ordering $>$ on $\Field[x_1,\ldots,x_n]$ (see~\cite{Cox}, Definition 1, p.55), which will not be defined explicitly.
We can reconstruct the monomial $x^\alpha=x_1^{\alpha_1}\cdots x_n^{\alpha_n}$ from the $n$-tuple of exponents $\alpha =(\alpha_1,\ldots,\alpha_n)\in \Zz^n_{\geq0}$. This establishes a one-to-one correspondence between the monomials in $\Field[x_1,\ldots,x_n]$ and $\Zz^n_{\geq0}$. Any ordering $>$ we establish on the space $\Zz^n_{\geq0}$
will give us an ordering on monomials: if $\alpha > \beta$ according to this ordering, we will also say that $x^\alpha > x^\beta$. The two monomial orderings that we use in this paper are the lexicographic order $>_\lex$ and the graded lexicographic ordering $>_\grlex$. 

\begin{definition}\label{def:lex and grlex} Let $\alpha =(\alpha_1,\ldots,\alpha_n),\beta=(\beta_1,\ldots,\beta_n)\in \Zz^n_{\geq0}$ and $|\alpha| = \sum_{i=1}^n\alpha_i$, $|\beta| = ~\sum_{i=1}^n\beta_i$.
  \begin{enumerate}[(i)]
      \item We say $\alpha>_\lex \beta$ if, in the vector difference $\alpha -\beta \in \Zz^n$, the left most nonzero entry is positive. We will write $x^\alpha>_\lex x^\beta$ if $\alpha>_\lex \beta$. 
      \item We say $\alpha>_\grlex \beta$ if $|\alpha| >|\beta|$, or $|\alpha| =|\beta|$ and $\alpha>_\lex \beta$.

  \end{enumerate}
\end{definition}

\begin{definition}
  For any $\alpha=(\alpha_1,\cdots,\alpha_n)\in \Zz^n_{\geq0}$ let $x^\alpha\mydef \prod_{i=1}^{n}x_i^{\alpha_i}$. Let $f= \sum_{\alpha} a_{\alpha}x^\alpha$ be a nonzero polynomial in $\Field[x_1,\ldots,x_n]$ and let $>$ be a monomial order.
  \begin{enumerate}[(i)]
    
    \item The \textbf{\emph{multidegree}} of $f$ is $\multideg(f)\mydef \max(\alpha\in \Zz^n_{\geq0}:a_\alpha\not = 0)$.
    \item The \textbf{\emph{degree}} of $f$ is deg$(f)=|\multideg(f)|$. In this paper, this is always according to \grlex order.
    \item The \textbf{\emph{leading coefficient}} of $f$ is $\LC(f)\mydef a_{\multideg(f)}\in \Field$.
    \item The \textbf{\emph{leading monomial}} of $f$ is $\LM(f)\mydef x^{\multideg(f)}$ (with coefficient 1).
    \item The \textbf{\emph{leading term}} of $f$ is $\LT(f)\mydef \LC(f)\cdot \LM(f)$.
  \end{enumerate}
\end{definition}
%

The concept of \emph{reduction}, also called \emph{multivariate division} or \emph{normal form computation}, is central to \GB basis theory. It is a multivariate generalization of the Euclidean division of univariate polynomials.

\begin{definition}\label{def:reduction}
Fix a monomial order and let $G=\{g_1,\ldots,g_t\}\subset \Field[x_1,\ldots,x_n]$. Given $f\in \Field[x_1,\ldots,x_n]$, we say that \emph{\textbf{$f$ reduces to $r$ modulo $G$}}, written
$f\rightarrow_G r$,
if $f$ can be written in the form
$f=A_1g_1+\dots+A_t g_t+r$ for some $A_1,\ldots,A_t,r\in \Field[x_1,\ldots,x_n]$,
such that:
\begin{enumerate}[(i)]
  \item No term of $r$ is divisible by any of $\LT(g_1),\ldots,\LT(g_t)$.
  \item Whenever $A_i g_i\not=0$, we have $\multideg(f)\geq \multideg(A_ig_i)$.
\end{enumerate}
The polynomial remainder $r$ is called a \emph{\textbf{normal form of $f$ by $G$}} and will be denoted by $f|_G$.
\end{definition}

A normal form of $f$ by $G$, i.e. $f|_G$, can be obtained by repeatedly performing the following until it cannot be further applied: choose any $g\in G$ such that $\LT(g)$ divides some term $t$ of $f$ and replace $f$ with $f-\frac{t}{\LT(g)}g$. Note that the order we choose the polynomials $g$ in the division process is not specified.

In general a normal form $f|_G$ is not uniquely defined.
Even when $f$ belongs to the ideal generated by $G$, i.e. $f\in \Ideal{G}$, it is not always true that $f|_G=0$.
\begin{example}
  Let $f=xy^2-y^3$ and $G=\{g_1,g_2\}$, where $g_1=xy-1$ and $g_2=y^2-1$. Consider the graded lexicographic order (with $x>y$) and note that
    $f = y\cdot g_1 - y\cdot g_2 + 0$ and
    $f = 0\cdot g_1 + (x-y)\cdot g_2 + x-y$.
\end{example}
This non-uniqueness is the starting point of \GB basis theory.
\begin{definition}
Fix a monomial order on the polynomial ring $\Field[x_1,\ldots,x_n]$. A finite subset $G = \{g_1,\ldots, g_t\}$ of an ideal $\I \subseteq \Field[x_1,\ldots,x_n]$ different from $\{0\}$ is said to be a \emph{\textbf{\GB basis}} (or \emph{\textbf{standard basis}}) if
$\langle \LT(g_1),\ldots, \LT(g_t)\rangle = \langle \LT(\I)\rangle$, where we denote by $\langle \LT(\I)\rangle$ the ideal generated by the elements of the set $\LT(\I)$ of leading terms of nonzero elements of $\I$.
\end{definition}
\begin{definition}\label{def:redGB}
A \textbf{\emph{reduced \GB basis}} for a polynomial ideal $\I$ is a \GB basis $G$ for $\I$ such that:
\begin{enumerate}[(i)]
  \item $\LC(g)= 1$ for all $g \in G$.
  \item For all $g \in G$, $g$ cannot reduce any other polynomial from $G$, i.e $\reduce{f}{g}=f$ for every $f\in G\setminus\{g\}$.
\end{enumerate}
\end{definition}
It is known (see~\cite{Cox}, Theorem~5, p.93) that for a given monomial ordering, a polynomial ideal $\I\not=\{0\}$ has a reduced \GB basis (see Definition~\ref{def:redGB}), and the reduced \GB basis is unique. 
\begin{proposition}[\cite{Cox}, Proposition~1, p.83]\label{th:gbprop}
Let $\I\subset \Field[x_1,\dots,x_n]$ be an ideal and let $G=\{g_1,\ldots,g_t\}$ be a \GB basis for $\I$. Then given $f\in \Field[x_1,\dots,x_n]$, $f$ can be written in the form
$f=A_1g_1+\dots+A_t g_t+r$ for some $A_1,\ldots,A_t,r\in \Field[x_1,\ldots,x_n]$,
such that:
\begin{enumerate}[(i)]
  \item No term of $r$ is divisible by any of $\LT(g_1),\ldots,\LT(g_t)$.
  \item Whenever $A_i g_i\not=0$, we have $\multideg(f)\geq \multideg(A_ig_i)$.
  \item There is a unique $r\in \Field[x_1,\dots,x_n]$.
\end{enumerate}
In particular, $r$ is the remainder on division of $f$ by $G$ no matter how the elements of $G$ are listed when using the division algorithm.
\end{proposition}
\begin{corollary}[\cite{Cox}, Corollary~2, p.84]\label{th:imp}
  Let $G=\{g_1,\ldots,g_t\}$ be a \GB basis for $\I\subseteq \Field[x_1,\dots,x_n]$ and let $f\in \Field[x_1,\dots,x_n]$. Then $f\in \I$ if and only if the remainder on division of $f$ by $G$ is zero.
\end{corollary}
\begin{definition}\label{def:pdiv}
We will write $\reduce f F$ for the remainder of $f$ by the ordered $s$-tuple $F=(f_1,\ldots,f_s)$. If $F$ is a $\GB$ basis for $\spn{f_1,\dots,f_s}$, then we can regard $F$ as a set (without any particular order) by Proposition~\ref{th:gbprop}.
\end{definition}
%

%
The ``obstruction'' to $\{g_1,\ldots, g_t\}$ being a \GB basis is the possible occurrence of polynomial combinations of the $g_i$ whose leading terms are not in the ideal generated by the $\LT( g_i)$. One way (actually the only way) this can occur is if the leading terms in a suitable combination cancel, leaving only smaller terms. The latter is fully captured by the so called $S$-polynomials that play a fundamental role in \GB basis theory.
\begin{definition}\label{def:spoly}
  Let $f,g\in \Field[x_1,\ldots,x_n]$ be nonzero polynomials.
    If $\multideg(f) =\alpha$ and $\multideg(g)= \beta$, then let $\gamma=(\gamma_1,\ldots,\gamma_n)$, where $\gamma_i = \max(\alpha_i,\beta_i)$ for each $i$. We call $x^\gamma$ the \emph{\textbf{least common multiple}} of $\LM(f)$ and $\LM(g)$, written $x^\gamma = \LCM(\LM(f),\LM(g))$.
    The \emph{\textbf{$S$-polynomial}} of $f$ and $g$ is the combination $S(f,g) = \frac{x^\gamma}{\LT(f)}\cdot f - \frac{x^\gamma}{\LT(g)}\cdot g$.
\end{definition}
The use of $S$-polynomials to eliminate leading terms of multivariate polynomials generalizes the row reduction algorithm for systems of linear equations. If we take a system of homogeneous linear equations (i.e.: the constant coefficient equals zero), then it is not hard to see that bringing the system in triangular form yields a \GB basis for the system.
\begin{theorem}[\textbf{Buchberger's Criterion}] (See e.g. \cite{Cox}, Theorem 3, p.105)\label{th:crit}
A basis $G=\{g_1,\ldots,g_t\}$ for an ideal $\I$ is a \GB basis if and only if $S(g_i,g_j)\rightarrow_G 0$ for all $i\not=j$.
\end{theorem}
%
By Theorem~\ref{th:crit} it is easy to show whether a given basis is a \GB basis. Indeed, if $G$ is a \GB basis then given $f\in \Field[x_1,\dots,x_n]$, $f|_G$ is unique and it is the remainder on division of $f$ by $G$, no matter how the elements of $G$ are listed when using the division algorithm.

Furthermore, Theorem~\ref{th:crit} leads naturally to an algorithm for computing \GB bases for a given ideal $\I=\langle f_1,\ldots,f_s \rangle$: start with a basis $G=\{f_1,\ldots,f_s\}$ and for any pair $f,g\in G$ with $S(f,g)|_G\not= 0$ add $S(f,g)|_G$ to $G$.
This is known as Buchberger's algorithm~\cite{BUCHBERGER2006475} (for more details see Algorithm~\ref{Algo:buchbergerAlgo} in Section~\ref{sect:buchbergerAlgo}).

Note that Algorithm~\ref{Algo:buchbergerAlgo} is non-deterministic and the resulting \GB basis in not uniquely determined by the input. This is because the normal form $S(f,g)|_G$ (see Algorithm~\ref{Algo:buchbergerAlgo}, line~\ref{eq:reminder}) is not unique as already remarked.
We observe that one simple way to obtain a deterministic algorithm (see \cite{Cox}, Theorem~2, p. 91) is to replace $h:=S(f,g)|_G$ in line~\ref{eq:reminder} with $h:=\reduce {S(f,g)} G$ (see Definition~\ref{def:pdiv}), where in the latter $G$ is an ordered tuple. However, this is potentially dangerous and inefficient. Indeed, there are simple cases where the combinatorial growth of set $G$ in Algorithm~\ref{Algo:buchbergerAlgo} is out of control very soon.



\subsubsection{Construction of \GB Bases.}\label{sect:buchbergerAlgo}
Buchberger's algorithm~\cite{BUCHBERGER2006475} can be formulated as in Algorithm~\ref{Algo:buchbergerAlgo}. 
\begin{algorithm}
\caption{Buchberger's Algorithm}
\label{Algo:buchbergerAlgo}
\begin{algorithmic}[1]
\STATE \textbf{Input}: A finite set $F=\{f_1,\ldots,f_s\}$ of polynomials
\STATE \textbf{Output}: A finite \GB basis $G$ for $\spn{f_1,\ldots,f_s}$
\STATE $G:= F$\label{eq:startingG}
\STATE $C:= G\times G$
\WHILE{$C\not = \emptyset$}
\STATE Choose a pair $(f,g)\in C$
\STATE $C:=C\setminus \{(f,g)\}$
\STATE $h:=S(f,g)|_G$ \label{eq:reminder}
\IF {$h\not=0$}\label{eq:Buch_cond}
\STATE $C:= C\cup (G\times \{h\})$
\STATE $G:= G\cup \{h\}$\label{eq:addspoly}
\ENDIF
\ENDWHILE
\STATE Return G
\end{algorithmic}
\end{algorithm}
The pairs that get placed in the set $C$ are often referred to as \emph{critical pairs}. Every newly added reduced $S$-polynomial enlarges the set $C$. If we use $h:=\reduce {S(f,g)} G$ in line~\ref{eq:reminder} then there are simple cases where the situation is out of control. This combinatorial growth can be controlled to some extent be eliminating unnecessary critical pairs.

\end{document}